\documentclass[11pt]{article}
\usepackage{fullpage}
\usepackage{apacite}
\usepackage{comment}
\usepackage{graphicx,amsmath,amsfonts,amsthm}

\newtheorem{theorem}{Theorem}

\newtheorem{proposition}[theorem]{Proposition}

\newtheorem{definition}[theorem]{Definition}

\newcommand{\p}{\ensuremath{\mathrm{P}}}
\newcommand{\np}{\ensuremath{\mathrm{NP}}}

\newcommand{\calC}{{\cal C}}
\newcommand{\calE}{{\cal E}}

\newcommand{\dod}{\mathrm{swap}}
\newcommand{\discr}{\mathrm{discr}}
\newcommand{\plur}{\mathrm{plur}}

\newcommand{\reals}{\mathbb R}

\newcommand{\pref}{o} 
\newcommand{\vecpref}{O}
\newcommand{\electionsystem}{{\cal{R}}}
\newcommand{\R}{\electionsystem}

\title{Rationalizations of Condorcet-Consistent Rules via  Distances of Hamming Type}

\author{{Edith Elkind}\\
      School of ECS\\
      University of Southampton, UK
      and\\
	Division of Mathematical Sciences\\
	Nanyang Technological University\\
	Singapore
        \and 
        {Piotr Faliszewski}\\
        Department of Computer Science \\
        AGH University of Science\\ and Technology,
        Krak\'ow\\ Poland
	\and 
        {Arkadii Slinko}\\
	Department of Mathematics\\
	University of Auckland \\
        New Zealand
	}
\date{\today}

\begin{document}

\maketitle

\begin{abstract}
The main idea of the {\em distance rationalizability} approach to view the voters'
preferences as an imperfect approximation to some kind of
consensus is deeply rooted in social choice literature.  It allows one to define (``rationalize'')
voting rules via a consensus class of elections  and a distance: a candidate is said to be an election winner 
if she is ranked first in one of the nearest (with respect to the given 
distance) consensus elections.  It is known that many classic voting rules can be distance rationalized. 
In this paper, we provide new results on distance rationalizability
of several Condorcet-consistent voting rules. In particular, we distance rationalize Young's rule and Maximin rule using distances similar to the Hamming distance. We show that the claim that  
Young's rule can be rationalized by the Condorcet consensus class and the Hamming distance is incorrect; 
in fact, these consensus class and distance yield a new rule which has not been studied before. 
We prove that, similarly to Young's rule, this new rule has a computationally hard winner determination problem.
\end{abstract}

\section{Introduction}\label{sec:intro}

The problem of defining what is meant by an electoral consensus
has been a particulariy contentious one.
Condorcet approached this problem from the point of view of 
pairwise comparisons. He suggested that, if an alternative obtains a simple majority over any other alternative, then it should win the election. This principle is known as {\em Condorcet rule} and the winner as {\em Condorcet alternative}. Despite all its attractiveness this principle has a major drawback: a Condorcet alternative does not always exist. Various methods of extending Condorcet rule to all elections have been proposed; one of the most attractive ways to do so was suggested by \citeA{you:j:extending-condorcet}. He viewed the problem of social choice as a problem in 
pattern recognition. In cases where the ``pattern'' of consensus is unclear---that is, a 
Condorcet alternative does not exist---he suggested to use the 
majority principle and to search for the largest subset of voters for which the pattern is clear and a Condorcet alternative  exists. 
\citeA{nit:j:closeness} expressed similar ideas with respect to the unanimity rule. 

Distance rationalizability is a framework formalizing this research direction. 
The idea of this framework is to view the voters'
collection of preferences, or a {\em preference profile}, as an imperfect approximation to some kind of
consensus. Identifying the ``closest'' consensus profile  we ``recognize the pattern.'' The winner is then the most preferred
candidate in this closest consensus profile. A voting
rule can be defined by picking a particular notion of a consensus
and a particular notion of closeness. This closeness must be measured by a distance function since violations of the triangle inequality 
may lead to undesirable effects. These ideas has been explored by several
authors~\cite{bai:j:distance-rationalisation,kla:j:copeland-distance,kla:j:borda-vs-condorcet}
under a variety of names;
 a fairly comprehensive list of distance-rationalizability results  
is provided by~\citeA{mes-nur:b:distance-realizability}. 

A surprisingly large number of voting rules have been already classified as distance rationalizable. \citeA{nit:j:closeness} distance rationalized Plurality and Borda,  \citeA{mes-nur:b:distance-realizability}, among other rules, provide distance rationalizations of Veto, Cope\-land, Slater, and STV,
and~\citeA{elk-fal-sli:c:votewise-dr} show that all scoring rules, as well
as the Bucklin rule,\footnote{Also known as majoritarian compromise.} are also
distance rationalizable.  Some rules, like Dodgson's rule or Kemeny rule, have been initially 
defined in terms of a consensus class and a distance so no additional rationalization was required.  
Effectively the idea has developed into a project of classification 
of existing voting rules by two parameters: a consensus class and a distance.

However, up to date this classification  has contained some gaps. Paradoxically enough one of them relates to 
Young's rule which appeared to be notoriously tricky to rationalize. 
\citeA{mes-nur:b:distance-realizability} claim that Young's rule obtains if we employ the Condorcet 
consensus class and the Hamming distance over the profiles, however this is not true. 

The first goal of this paper is to show that the statement of Meskanen and Nurmi is wrong. In fact,  
using the Condorcet  consensus class and the Hamming distance we obtain a new rule which is different from 
Young's rule and also any other known rule. We call it {\em voter replacement rule} until a better name for it is found. 
We study this rule and prove that, similarly to Young's rule, it 
has a computationally hard winner determination problem.
 The second goal is to provide a correct distance rationalizability results  both for Young's rule and Maximin rule filling the existing gaps.  Creating  distances for  these rules appeared to be more involved than one's intuition might initially suggest.

Our paper is organized as follows. In Section~\ref{sec:prelim} we
formally describe our model of elections and the distance
rationalizability framework, tailored to the case of Condorcet
consensus. Then, in Section~\ref{sec:condorcet} we show that Young's
rule and Maximin are both distance rationalizable via distances 
which are, in spirit, similar to Hamming distance.
We show that the rule obtained by Hamming distance itself is 
different from Young's rule. 
We prove that  the winner determination problem
for this new rule is computationally hard.

We discuss our
results and present further research directions in
Section~\ref{sec:conclusions}. In the appendix we very briefly
describe fundamental notions of the computational
complexity theory.

\section{Preliminaries}\label{sec:prelim}

An {\em election} $E$ is a triple 
$(C,V, \vecpref)$, where $C = \{c_1, \ldots,
c_m\}$ is a set of {\em candidates}, $V=\{v_1, \ldots, v_n\}$ is a
set of {\em voters}, and $\vecpref=(\pref_1, \dots, \pref_n)$
is a {\em preference profile}, i.e., a vector of {\em preference orders}
of the voters in $V$. For each $i=1, \dots, n$,   
$\pref_i$ is a strict total order over
the candidates in $C$. For readability, we sometimes write $\succ_i$
instead of $o_i$. 
For example, given a candidate set $C = \{c_1,
c_2, c_3\}$, and a voter $v_i$ that likes $c_2$ best, then $c_1$, and then
$c_3$, we write $c_2 \succ_i c_1 \succ_i c_3$.
We remark that it is common to identify the voter set $V$
with the preference profile $(\pref_1, \dots, \pref_n)$.
However, since in this paper we will consider actions 
that modify the set of voters, it will be more convenient
to treat $V$ and $(\pref_1, \dots, \pref_n)$ as two distinct objects. 

A \emph{voting rule} $\R$ (or, more precisely, a social choice
correspondence $\R$) is a function that given an election $E = (C,V, \vecpref)$
outputs a set $\R(E)\subseteq C$ of {\em winners} of the
election. 
Note that we do not require $|\R(E)|=1$. Indeed, there are
cases where, e.g., due to symmetry, it is impossible to declare a
single winner, in which case we may have $\R(E)=\emptyset$ or
$|\R(E)|>1$.  In practice, one may then need to use a {\em draw
  resolution rule}, which can be either deterministic (e.g.,
lexicographic) or randomized (e.g., a fair coin toss); however, in the
rest of this paper we will ignore this issue.  Perhaps the best known
voting rule is the Plurality rule $\R_\plur$, which elects those
candidates who are ranked first by the largest number of voters.

We say that a candidate $c_i$ is a {\em Condorcet winner} in an
election $E = (C,V, \vecpref)$ if for each $c_j \in C$, $c_i \neq c_j$, a strict
majority of voters prefers $c_i$ to $c_j$. While not every election
has a Condorcet winner, the notion is so appealing that many
rules---so-called Condorcet-consistent rules---are designed to select
the Condorcet winner if it exists. For example, Dodgson's rule
selects those candidates who can be made Condorcet winners by the least
number of swaps of adjacent candidates in the preference orders of the
voters.

Intuitively, a preference profile corresponds to a {\em consensus}
among the voters when there exists an alternative that is clearly
better from the collective point of view than any other one.  For
example, one could consider \emph{strongly unanimous} profiles, where
all voters rank candidates identically, or \emph{weakly unanimous}
profiles, where all voters agree on the top-ranked candidate. In
either case it is obvious that the top-ranked candidate is clearly
better than any other one. Throughout this paper, 
we consider a weaker type of consensus, which is inspired 
by the idea that a Condorcet winner, when one exists, 
presents an acceptable compromise between different voters' preferences. 
That is, we say that an election is a {\em consensus election} if it has a
Condorcet winner; we denote the set of all such elections by $\calC$.
For technical reasons, we assume that $\calC$ does not contain
an election with an empty set of voters.

Given a set $X$, we say that a function $d \colon X \times X
\rightarrow \reals \cup \{+\infty\}$ is a {\em distance} (or {\em
  metric}) over $X$ if for each $x,y \in X$ it satisfies the following
four axioms:
\begin{itemize}
\item[(1)] 
 $d(x, y)\ge 0$ (non-negativity),
\item[(2)] 
 $d(x,y) = 0$ if and only if $x = y$ (identity of indiscernibles),
\item[(3)] 
 $d(x,y) = d(y,x)$ (symmetry), and
\item[(4)] 
 for each $z \in X$, $d(x,y) \leq d(x,z) + d(z,y)$ (triangle inequality).
\end{itemize}
In what follows, the elements of the set $X$ will usually be either
voters (i.e., preference orders) or elections.

Any distance $d(\pref, \pref')$ over voters 
with preferences over a candidate set $C$
can be extended to a
distance $\widehat{d}(E^1, E^2)$ over elections $E^1=(C, V, \vecpref^1)$
and $E^2=(C, V, \vecpref^2)$ 
with $\vecpref^1 = (\pref^1_1, \dots, \pref^1_n)$, 
     $\vecpref^2 = (\pref^2_1, \dots, \pref^2_n)$
by setting $\widehat{d}(E^1, E^2)=\sum_{i=1}^nd(\pref^1_i, \pref^2_i)$.
Clearly, $\widehat{d}$
satisfies all distance axioms as long as $d$ does.\footnote{We point
the reader to the work of~\citeA{elk-fal-sli:c:votewise-dr} for an extensive discussion
of distance rationalizability via distances of this type.}

We now provide two
examples of distances defined over pairs of voters with preferences
over a set of candidates $C$.  
Our first example is the {\em discrete
  distance} $d_\discr(\pref, \pref')$, given by $d_\discr(\pref, \pref')=1$ if 
$\pref\neq \pref'$
and $d_\discr(\pref, \pref')=0$ otherwise. Clearly, the corresponding distance
over elections $\widehat{d}_\discr(E^1, E^2)$ is equivalent to the {\em
  Hamming distance} $d_H(E^1, E^2)$, which is defined as $d_H(E^1,
E^2)=|\{i\mid \pref^1_i\neq \pref^2_i\}|$.  Our second example is the {\em Dodgson
  distance}, or {\em swap distance}, $d_\dod(\pref,\pref')$, defined as
$d_\dod(\pref,\pref')=|\{(c_1, c_2)\in C^2\mid c_1\,\pref\, c_2, c_2\,\pref'\, c_1\}|$.
It is not hard to check that both the Dodgson distance and the discrete
distance (and hence the Hamming distance) satisfy the distance axioms
listed above.  (Note that, formally, both of these distances are defined
only for pairs of elections with the same candidate sets and the same
voter sets; if either of these conditions is not met, we assume
that the distance is $\infty$.)  

We are now ready to define distance rationalizability.
The following two definitions are specialized to rationalizability
with respect to Condorcet consensus, but can be adapted
to apply to other consensus classes in a straightforward manner.

\begin{definition}\label{def:score}\label{def:distance-voting}
  Let $d$ be a distance over elections. We define the {\em
    $(\calC,d)$-score} of a candidate $c_i$ in an election $E$ to be
  the distance (according to $d$) between $E$ and a closest election
  $E'$ where $c_i$ is the Condorcet winner.  The set of
  {\em $(\calC,d)$-winners} of an election $E = (C,V, \vecpref)$ consists of those
  candidates in $C$ whose $(\calC,d)$-score is smallest.
\end{definition}

\begin{definition}\label{def:dr}
  A voting rule $\R$ is {\em distance-rationalizable} via Condorcet
  consensus and a distance $d$ over elections, or
  {\em $(\calC,d)$-rationalizable}, if for each election $E$, a candidate
  $c$ is an $\R$-winner of $E$ if and only if she is a
  $(\calC,d)$-winner of $E$.
\end{definition}

For example, Dodgson's rule is $(\calC,\widehat{d}_\dod)$-rationalizable.  This result follows
directly from the definition of Dodgson's rule and witnesses that at
least some voting rules are naturally represented within the distance rationalizability framework.

\section{Main Results}\label{sec:condorcet}
In this section we present our results on voting rules that can be
rationalized with respect to the Condorcet consensus
via Hamming-type distances that correspond to adding, deleting, and
replacing voters.  It is important to have in mind that, when we speak, for example, about deleting voters, 
no voters are actually being deleted. They are just excluded from consideration 
in a search of a maximal subgroup in the electorate that possesses a Condorcet winner.

To begin, observe that, given an election $E=(C, V, \vecpref)$ with $|V|=n$, we
can make any candidate $c\in C$ the Condorcet winner by adding at most
$n+1$ voters that rank $c$ first. Similarly, we can make $c$ the
Condorcet winner by replacing at most $\lfloor n/2\rfloor+1$ voters in
$V$ with voters that rank $c$ first. While not every candidate can be
made the Condorcet winner by voter deletion---for example, if a
candidate is ranked last by all voters, he will not become the Condorcet
winner no matter how many voters we delete---it is still the case
that, if at least one voter ranks a given candidate first, this
candidate can be made the Condorcet winner by removing at most $n-1$
voters.  Thus, for each candidate $c$ we can define her score with respect to
each of these operations as the number of voters that need to be
inserted, replaced, or removed, respectively, to make $c$
the Condorcet winner (for deletion, some candidates will have a score
of $+\infty$).  We will refer to these scores as the {\em insertion
  score}, the {\em replacement score} and the {\em deletion score},
respectively.  Intuitively, for each of these scores, the candidates
with a lower score are closer to being the consensus winners than the
candidates with a higher score, so each of these scores can be used to
define a voting rule.

In fact, there is a well-known voting rule that is defined in these
terms, namely, Young's rule, which elects the candidates with the
lowest deletion score.  Thus, it is natural to ask if the two other
scores defined above, i.e., the replacement score and the insertion
score, also correspond to well-known voting rules. Another interesting
question is whether all three of these scores can be transformed into
distances, i.e., whether the corresponding voting rules are
distance-rationalizable with respect to the Condorcet consensus; 
observe that this issue is more complicated than might
appear at the first sight, since we have to satisfy the symmetry axiom.
Providing answers to these questions is the main contribution
of our paper.

We will first answer the second question by showing how to transform
each of our three scores into a distance. The easiest case is that of
the replacement score.  Formally, given an election $E=(C, V, \vecpref)$, the
{\em replacement score} $s_r(c)$ of a candidate $c\in C$ is the
smallest value of $k$ such that there exists an election $E=(C, V, \vecpref')$
obtained by changing the preferences of exactly $k$ voters in $V$ in
which $c$ is the Condorcet winner; as argued above, $s_r(c)\le \lfloor
n/2\rfloor+1$ for all $c\in C$.  It is immediate that the replacement
score of any $c\in C$ is exactly the Hamming distance from $E$ to the
closest election over the set of candidates $C$ in which $c$ is the
Condorcet winner.  Thus, the corresponding voting rule is $(\calC,
d_H)$-rationalizable. We will refer to this rule as the voter
replacement rule.  We postpone the discussion of whether this rule is
equivalent to any voting rule considered in the literature till the
end of the section.

The {\em insertion score} $s_i(c)$ of a candidate $c\in C$ in an
election $E=(C, V, \vecpref)$ is defined as the smallest number $k\ge 0$ such
that there exists a set of voters $V'$, $|V'|=k$, 
with a preference profile $\vecpref'$ over $C$ such that $c$ is
the Condorcet winner in $E'=(C, V\cup V', \vecpref\circ\vecpref')$, 
where $\vecpref\circ\vecpref'$ denotes the concatenation of the preference
profiles $\vecpref$ and $\vecpref'$.  
Similarly, the {\em
  deletion score} $s_d(c)$ of a candidate $c\in C$ in an election
$E=(C, V, \vecpref)$ is defined as the smallest number $k\ge 0$ such that there
exists a subset of voters $V'\subseteq V$, $|V'|=k$, such that $c$ is
the Condorcet winner in $E'=(C, V\setminus V', \vecpref\setminus\vecpref')$, 
and $+\infty$ if $c$ cannot be made the Condorcet winner in this manner.
Here, $\vecpref\setminus\vecpref'$ denotes the preference profile obtained
from $\vecpref$ by deleting the preference orders of voters in $V'$.

Now, it is easy to see that both the insertion score and the deletion
score naturally correspond to quasidistances, i.e., mappings that
satisfy non-negativity, identity of indiscernibles and the triangle
inequality, but not symmetry.  Indeed, given two elections $E=(C, V, \vecpref)$
and $E=(C, V', \vecpref')$ over the same set of candidates $C$, 
we can define a
function $d'_i(E, E')$ by setting $d'_i(E, E')=k$ if 
$\pref_i=\pref'_i$ for each $v_i\in V\cap V'$, $V\subseteq V'$ and
$|V'\setminus V|=k$, and $d'_i(E, E')=+\infty$
otherwise. 
Similarly, we can define $d'_d(E, E')$ by setting $d'_d(E,
E')=k$ if $\pref_i=\pref'_i$ for each $v_i\in V\cap V'$,
$V'\subseteq V$ and $|V\setminus V'|=k$, and
$d'_d(E, E')=+\infty$ otherwise. It is not hard to verify that both
$d'_i$ and $d'_d$ are quasidistances. Moreover, for each candidate in
$C$ his insertion score $s_i(c)$ is equal to the $d'_i$-distance from
$E$ to the nearest (with respect to $d'_i$) election in $\calC$ in
which $c$ is the Condorcet winner.  Similarly, $c$'s deletion score
$s_d(c)$ is equal to the $d'_d$-distance from $E$ to the nearest (with
respect to $d'_d$) election in $\calC$ in which $c$ is the Condorcet
winner.  We will now show that we can replace both of these
quasidistances with true distances.

For $d'_i$ the solution is simple: we can make $d'_i$ symmetric by
allowing ourselves to delete voters as well as to add voters, as,
intuitively, deleting a voter is never more useful than adding a
voter. Formally, given two elections $E=(C, V, \vecpref)$ and 
$E=(C, V', \vecpref')$ over
the same set of candidates $C$, we set $d_i(E, E')=|V\setminus
V'|+|V'\setminus V|$ if $\pref_i=\pref'_i$ for all $v_i\in V\cap V'$
and $d_i(E, E')=+\infty$ otherwise.
Clearly, $d_i$ is a distance. Moreover,
we will now show that for our purposes it is indistinguishable from
$d'_i$.
\begin{proposition}\label{prop:dr-ins}
  Consider an election $E=(C, V, \vecpref)$, a candidate $c\in C$, and a $k>0$.
  Then there exists an election $E^1=(C, V^1, \vecpref^1)\in\calC$ such that $c$
  is the Condorcet winner of $E^1$ and $d'_i(E, E^1)\le k$ if and only
  if there exists an election $E^2=(C, V^2, \vecpref^2)\in\calC$ such that $c$ is
  the Condorcet winner of $E^2$ and $d_i(E, E^2)\le k$.
\end{proposition}
\begin{proof}
  The ``only if'' direction is immediate: if $d'_i(E, E^1)\le k$, then
  $d_i(E, E^1)\le k$, so we can set $E^2=E^1$. For the ``if''
  direction, suppose that $E^2$ has been obtained from $E$ by deleting
  a subset of voters $V'\subseteq V$, $|V'|=k_1$, and adding a set
  of voters $V''$ with a preference profile $\vecpref''$, 
  $|V''|=k_2$. Now, consider an election $E^3$
  obtained from $E$ by first adding the voters in $V''$ and then
  adding another $k_1$ voters that rank $c$ first. Clearly, $d'_i(E,
  E^3)=d_i(E, E^2)\le k$. We will now show that $c$ is the Condorcet
  winner in $E^3$. Indeed, fix an arbitrary voter $c'\in C$.  Suppose
  that in $(C, V\cup V'', \vecpref\circ\vecpref'')$ 
  there are $x$ voters that prefer $c$ to
  $c'$ and $y$ voters that prefer $c'$ to $c$. Then in $E^2$ there are
  at most $x$ voters that prefer $c$ to $c'$ and at least $y-k_1$
  voters that prefer $c'$ to $c$.  Since $c$ is the Condorcet winner
  of $E^2$, we have $x>y-k_1$.  Now, in $E^3$ there are $x+k_1$ voters
  that prefer $c$ to $c'$ and $y$ voters that prefer $c'$ to $c$. As
  we have argued that $x+k_1>y$, it follows that the majority of
  voters in $E^3$ prefer $c$ to $c'$. As this is true for any $c'\neq
  c$, it follows that $c$ is the Condorcet winner in $E^3$.  Moreover,
  $E^3$ has been obtained from $E$ by candidate insertion only, so we
  can set $E^1=E^3$.
\end{proof}
Clearly, we cannot use the same solution for $d'_d$. Indeed, the
argument above demonstrates that adding voters is more useful than
deleting voters. Thus, we need to construct a metric that makes it
expensive to add voters. As this metric has to be symmetric, a natural
approach would be to make the distance between two elections depend on
the number of voters in the larger of them, as well as on the
difference in the number of voters.  For example, 
given two elections $E=(C, V, \vecpref)$ and  $E'=(C, V', \vecpref')$, 
we could try to set
$$
d(E, E')=
\begin{cases}
\left||V|-|V'|\right|+(\max\{|V|, |V'|\})^2 
 		&\text{if $\pref_i=\pref'_i$ for each $i\in V\cap V'$}\\
+\infty &\text{otherwise}.
\end{cases}
$$
However, it turns out that this approach does not quite work: under
this metric, deleting $s_d(c)$ voters may still be more 
expensive than
first deleting some $s'<s_d(c)$ voters and then adding a few voters
that rank $c$ first.  To overcome this difficutly, we construct a
metric that makes it prohibitively difficult to do insertion and
deletion at the same time.

Formally, for any pair of elections $E=(C, V, \vecpref)$, 
$E'=(C, V', \vecpref')$ over the same set of candidates $C$
such that $\pref_i=\pref'_i$ for each $i\in V\cap V'$, 
we set $k = \left||V|-|V'|\right|$, $M = \max\{|V|, |V'|\}$, and let
$$
\overline{d}_d(E, E')=
\begin{cases}
0 				&\text{if }V=V'\\
2- \frac{1}{k+M^2+1} 	&\text{if }V\subset V'\text{ or }V'\subset V\\
+\infty &\text{otherwise}.
\end{cases}
$$
Also, we set $\overline{d}_d(E, E') =+\infty$ if
$\pref_i\neq\pref'_i$ for some $i\in V\cap V'$.
The function $\overline{d}_d(E, E')$ is not a metric, as it does not satisfy
the triangle inequality. However, we can use it to construct a metric
$d_d$ by setting 
$
d_d(E, E')=\min
\{\overline{d}_d(E, E_1)+\overline{d}_d(E_1, E_2)+\dots+\overline{d}_d(E_\ell, E')
\mid \ell\in{\mathbb N}, 
E_1, \dots, E_\ell\in{\calE}_C
\}, 
$ 
where $\calE_C$ denotes the set of all elections with the set of
candidates $C$.  Intuitively, $d_d(E, E')$ is the shortest path
distance in the graph whose vertices are elections in $\calE_C$, and
the edge lengths are given by $\overline{d}_d$. It is well known that
for any graph with non-negative edge lengths the shortest path
distance satisfies the triangle inequality; it should be clear that
$d_d$ satisfies all other axioms of a metric as well.  Observe that
for any two elections $E, E'\in{\calE}_C$ such that $E=(C, V, \vecpref)$,
$E'=(C, V', \vecpref')$ and $\pref_i=\pref'_i$ for each $i\in V\cap V'$, 
we have $d_d(E, E')< 2$ if $V\subseteq V'$ or
$V'\subseteq V$ and $d_d(E, E')>2$ otherwise.

We will now show that $d_d$ can be used to rationalize Young's
rule with respect to the Condorcet consensus.

\begin{proposition}\label{prop:dr-del}
  Consider an election $E=(C, V, \vecpref)$, $|V|=n$, and two candidates $c_1,
  c_2\in C$ such that $s_d(c_1)<+\infty$ or $s_d(c_2)<+\infty$.  For
  $i=1, 2$, let $d_i$ be the $d_d$-distance from $c_i$ to the closest
  (with respect to $d_d$) election over $C$ in which $c_i$ is a
  Condorcet winner, Then $s_d(c_1)<s_c(c_2)$ if and only if $d_1<
  d_2$.
\end{proposition}
\begin{proof} 
  Suppose first that $s_d(c_1)=k_1<+\infty$, $s_d(c_2)=k_2<+\infty$.
  Then one can obtain an election over $C$ in which $c_1$
  (respectively, $c_2$) is the Condorcet winner by deleting $k_1$
  (respectively $k_2$) voters from $E$; denote this election by $E_1$
  (respectively, $E_2$).  We have $d_d(E, E_1)=2-\frac{1}{k_1+n^2+1}$,
  $d_d(E, E_2)=2-\frac{1}{k_2+n^2+1}$.  We claim that $d_1 = d_d(E,
  E_1)$.  Indeed, suppose that this is not the case, i.e., $d_d(E,
  E_1) > d_1$.  This means that there exists an election $E'=(C, V', \vecpref')$
  such that $c_1$ is the Condorcet winner of $E'$ and $d_d(E, E')<
  d_d(E, E_1)$.  As $E'$ cannot be obtained from $E$ by deleting
  voters, it holds that $V'\not\subseteq V$.  Now, if also
  $V\not\subseteq V'$, we immediately obtain $d_d(E, E')>2$, a
  contradiction with $d_d(E, E_1)<2$. Hence, it must be the case that
  $V\subset V'$, so $|V'|\ge n+1$, and we have $d_d(E, E')\ge
  2-\frac{1}{2+(n+1)^2}$. On the other hand, we have $k_1\le n-1$,
  which implies $d_d(E, E_1)\le 2-\frac{1}{n-1+n^2+1}$.  As
  $2-\frac{1}{2+(n+1)^2}> 2-\frac{1}{n-1+n^2+1}$, this gives a
  contradiction as well.  Similarly, we can show that $d_2 = d_d(E,
  E_2)$.  Hence, it follows that $k_1 < k_2$ if and only if $d_1 <
  d_2$.

  Now suppose that $s_d(c_1) < +\infty$, $s_d(c_2) = +\infty$ (the
  case $s_d(c_1) = +\infty$, $s_d(c_2) < +\infty$ is symmetric).  Then
  we have $d_1\le 2-\frac{1}{n+n^2}$, $d_2\ge 2-\frac{1}{2+(n+1)^2}$,
  since we cannot trasform $E$ into an election over $C$ in which
  $c_2$ is the Condorcet winner by candidate deletion only. Thus, in
  this case, too, $s_d(c_1) <s_d(c_2)$ if and only if $d_1 < d_2$.
\end{proof}
Since in any election there is at least one candidate $c$ with
$s_d(c)<+\infty$, Proposition~\ref{prop:dr-del} immediately implies
the following result.
\begin{theorem}
  Young's rule is $(\calC, d_d)$-rationalizable.
\end{theorem}

We now turn to the first of the two questions posed in the beginning
of this section.  We have observed that the voter deletion-based rule
is equivalent to Young's rule; the proof follows immediately from the
definitions of both rules.  We will now show that the voter
insertion-based rule is equivalent to another well-known rule, namely,
Maximin. Under Maximin, the score of each voter is the outcome of his
worst pairwise election. Formally, given an election $E=(C, V, \vecpref)$, 
for each $c_j\in C$ we set $s_M(c_j)=\min\{\#\{i: c_j\succ_i c_k\}\mid
c_k\in C\}$.  The winners are then the candidates $c$
with the highest Maximin score $s_M(c)$.
\begin{proposition}\label{prop:maximin}
  For any election $E=(C, V, \vecpref)$, $|V|=n$, and any candidate $c\in C$ we
  have $s_i(c) = n -2s_M(c)+1$, where $s_i(c)$ is the insertion score
  of $c$ and $s_M(c)$ is the Maximin score of $c$.
\end{proposition}
\begin{proof}
  Fix an election $E=(C, V, \vecpref)$, $|V|=n$, and a candidate $c_j\in C$.
  Set $t=s_M(c_j)$.  Let $c_k$ be one of $c_j$'s worst pairwise
  opponents, i.e., $|\{q: c_j\succ_q c_k\}|=t$. Now, if we add $n
  -2t+1$ voters that rank $c_j$ first, for any $c_\ell\neq c_j$ there
  are at most $n-t$ voters that rank $c_\ell$ above $c_j$ and at least
  $t+n-2t+1=n-t+1$ voters that rank $c_j$ above $c_\ell$, so $c_j$ is
  the Condorcet winner of the resulting election.  On the other hand,
  if we add at most $n-2t$ new voters to $E$, in the resulting
  election there will be at least $n-t$ voters that prefer $c_k$ to
  $c_j$ and at most $t+n-2t=n-t$ voters that prefer $c_j$ to $c_k$, so
  in this case $c_k$ prevents $c_j$ from becoming the Condorcet
  winner.
\end{proof}
Thus, the candidates with the highest Maximin score are exactly the
candidates with the lowest insertion score. Together with
Proposition~\ref{prop:dr-ins}, this implies the following result.
\begin{theorem}
  Maximin is $(\calC, d_i)$-rationalizable.
\end{theorem}

The situation with the voter replacement rule is more complicated.
\citeA{mes-nur:b:distance-realizability} claim that Young's rule
is $(\calC, d_H)$-rationalizable.  As we have argued that the voter
replacement rule is $(\calC, d_H)$-rationalizable, this would imply
that the voter replacement rule is equivalent to Young's rule, or, in
other words, deleting voters is equivalent to replacing voters.
However, it turns out that this is not true.

\begin{theorem}\label{thm:replacement-versus-young}
  There exists an election in which the voter replacement rule and
  Young's rule declare different candidates as winners.
\end{theorem}
\begin{proof}
  We construct an election $E=(C, V, \vecpref)$ with $C=\{a, b, c, d\}$ and
  $|V|=29$.  Among the first $5$ voters in $V$, there are $2$ voters
  with preference order $a\succ b \succ c\succ d$, $2$ voters with
  preference order $a\succ c \succ d\succ b$, and $1$ voter with
  preference order $a\succ b \succ d\succ c$.
 
  Further, there are $8$ voters with preferences $b \succ c \succ
  a\succ d$ ($b$-voters), $8$ voters with preferences $c \succ d \succ
  a\succ b$ ($c$-voters), and $8$ voters with preferences $d \succ b
  \succ a\succ c$ ($d$-voters).

  We summarize the numbers of voters that prefer $x$ to $y$ for $x,
  y\in\{a, b, c, d\}$ in the table below; we write $x>y:t$ to denote
  the fact that there are $t$ voters that prefer $x$ to $y$.
\begin{align*}
a>b: 13,\quad &b>a: 16,\quad &b>c: 19,\quad &c>b: 10\\
a>c: 13,\quad &c>a: 16,\quad &b>d: 11,\quad &d>b: 18\\
a>d: 13,\quad &d>a: 16,\quad &c>d: 20,\quad &d>c: 9
\end{align*}
Let us now compute $s_r(x)$ and $s_d(x)$ for $x\in\{b, c, d\}$.
Candidate $b$ wins pairwise elections against $a$ and $c$, but loses
to $d$ by $7$ votes. Hence, $s_d(b)\ge 8$.  On the other hand,
deleting $8$ votes is sufficient: indeed, deleting all $c$-voters
makes $b$ the Condorcet winner.  Thus, $s_d(b)=8$. For the same
reason, we need to replace at least $4$ voters to make $b$ the
Condorcet winner (each replacement reduces $d$'s margin of victory
over $b$ by at most $2$), and, indeed, replacing $4$ of the $c$-voters
with $4$ voters that rank $b$ first makes $b$ the Condorcet
winner. Hence, $s_r(b)=4$.  Similarly, $c$ loses the pairwise
election to $b$ by 9 votes, so we have $s_d(c)\ge 10$, $s_r(c)\ge 5$
(we can show that, in fact, $s_d(c)=10$ and $s_r(c)=5$, but this is
not needed for our proof), and $d$ loses the pairwise election to $c$
by 11 votes, so we have $s_d(d)\ge 12$, $s_r(d)\ge 6$.

Now, it is not hard to see that $s_r(a)\le 3$: after we replace one
$b$-voter, one $c$-voter and one $d$-voter with voters that rank $a$
first, for each $x=b, c, d$ we have $15$ voters that prefer $a$ to $x$
and $14$ voters that prefer $x$ to $a$.  Thus, we have $s_r(a)<s_r(x)$
for $x=b, c, d$. To complete the proof, we will now argue that
$s_d(a)>s_d(b)$. Specifically, we will show that $s_d(a)\ge 12$.

Indeed, it is clear that to make $a$ the Condorcet winner, it is never
optimal to delete any of the first five voters.  Now, suppose that we
can make $a$ the Condorcet winner by deleting a set $S$ of voters,
$|S|<12$.  Suppose first that $S$ contains at least $4$ voters of a
particular type (i.e., $b$-voters, $c$-voters, or $d$-voters); without
loss of generality, we can assume that $S$ contains $4$ $b$-voters.
After these voters have been deleted, $a$ loses to $d$ by $7$
votes, so we need to delete at least $8$ more voters, i.e., at least
$12$ voters altogether, a contradiction. Hence, we can now assume that
$S$ contains at most $3$ voters of each type. Next, suppose that $S$
contains exactly $3$ voters of some type; again, without loss of
generality we can assume that those are $b$-voters. After these voters
have been deleted, $a$ loses to $d$ by $6$ votes, so we have to
additionally delete at least $7$ other voters, i.e., at least $4$
voters of some other type, a contradiction.  Hence, $S$ contains at
most $2$ voters of each type. Now, consider an arbitrary voter in $S$;
without loss of generality we can assume that this is a $b$-voter.
After this voter has been deleted, $a$ loses to $d$ by $4$ votes, so
we need to additionally delete at least $5$ other voters, i.e., at
least $3$ voters of some other type, a contradiction. We conclude that
$s_d(a)\ge 12$.
\end{proof}  
In fact, the voter replacement rule, despite having a very natural
definition in terms of distances and consensuses, appears not to be
equivalent to any known voting rule. The only brief reference to this
rule that we could find in the literature is due to \citeA{fal-hem-hem:j:bribery}, where
the authors interpret this rule as a variant of the Dodgson rule, 
and show (\citeA{fal-hem-hem:j:bribery}, Theorem 5.7) that the problem of finding the
winners for the voter replacement rule is tractable (i.e., solvable in
polynomial time) under the assumption that the number of candidates is
fixed (in fact, their proof establishes something slightly stronger,
namely that the problem is fixed parameter tractable;
see~\cite{nie:b:invitation-fpt,dow-fel:b:parameterized} for
introduction to parameterized complexity). In contrast, we will now show
that without this assumption determining winners under the voter
replacement rule is computationally hard (this is also the case
for Young's rule; see the work of~\citeA{rot-spa-vog:j:young}).
\begin{theorem}\label{thm:nphard}
  Given an election $E=(C, V, \vecpref)$ and a candidate $p\in C$, it is $\np$-hard
  to decide if $p$ is a winner of $E$ under the voter replacement
  rule.
\end{theorem}

\begin{proof}
  We provide a many-one polynomial-time reduction from {\sc Vertex Cover}.  
 An instance of {\sc Vertex Cover} is
  given by a pair $(\Gamma=(X, Y); k)$ where $\Gamma$ is a graph with
   a vertex set $X$ and an edge set $Y$, and $k\in{\mathbb N}$. It is a
   ``yes''-instance if $\Gamma$ has a vertex cover of size at most $k$,
   and a ``no''-instance otherwise.  

We can assume that $|X|$ is divisible by $3$, i.e., $|X|=3q$ for some $q\in{\mathbb N}$, 
  and $|X|>3k+6$. Indeed, to show that such a restricted problem is
  $\np$-hard, we can reduce the unrestricted version of {\sc Vertex Cover}
  to it by adding a large enough ``star'' which is not connected to
  the rest of the graph. We can also assume that $\Gamma$
  has no isolated vertices, and therefore $|Y|\ge |X|/2$.

Given an instance $(\Gamma=(X, Y);k)$ of {\sc Vertex Cover}
with $|X|=N$, $|Y|=M$, we construct an election
$E=(C, V, \vecpref)$ as follows. Suppose that
$X=\{x_1, \dots, x_{N}\}$, $Y=\{y_1, \dots, y_{M}\}$.
Our election will have $M+5$ candidates $y_1, \dots, y_{M}, a, b, c, p, z$
and $2N-3$ voters. 
We identify the candidates $y_1, \dots, y_{M}$ with the corresponding
edges of $\Gamma$.
The first $N$ voters correspond to the vertices of $\Gamma$.
Specifically, for $i=1, \dots, N$, 
let $Y_i\subset Y$ be the set of edges incident to $x_i$. Then the voter 
$v_i$ ranks $a$, $b$, and $c$ on top, followed by the candidates in $Y_i$, 
followed by $p$, followed by the candidates in $Y\setminus Y_i$, followed by $z$.
We will specify the relative ordering of $a$, $b$ and $c$, 
as well as the relative ordering of the candidates in $Y_i$
and $Y\setminus Y_i$ in $v_i$'s vote later on.

All remaining $N-3$ voters rank all candidates in $Y$ above
$a$, $b$, $c$ and $p$. Among those voters, 
there are $k-2$ voters with preferences $a\succ p\succ b\succ c$, 
$k-2$ voters with preferences $b\succ p\succ c\succ a$, 
$k-2$ voters with preferences $c\succ p\succ a\succ b$ and
$N-3k+3$ voters that prefer $p$ to $a$, $b$, and $c$.
Furthermore, $N-k-1$ of the last $N-3$ voters rank $z$ first, while
the remaining $k-2$ voters rank $z$ last.

First, it is easy to see that $s_r(z)=k$.
Indeed, there are $N-k-1$ voters that rank $z$ first, 
and  $N+k-2$ voters that rank $z$ last.
Replacing $k$ of the voters that rank $z$ last
with ones that rank him first will make $z$ a majority winner, 
whereas if we replace less than $k$ voters, more than
half of the voters would still rank $z$ last.
We will now argue that (a) $s_r(p)\le k$ if and only if
$\Gamma$ has a vertex cover of size at most $k$; (b)
we can complete the specification of the voters'
preferences so that the replacement score of any candidate 
other than $p$ and $z$ is greater than $k$.

The first part is easy. Indeed, 
suppose that we can make
$p$ the Condorcet winner by replacing at most $k$ voters.
We claim that all voters that we replace are among the first $N$ voters.
Indeed, in $E$ there are $N+k-2$ voters that prefer $a$ to $p$
and $N-k-1$ voters that prefer $p$ to $a$. Thus, we have to replace
exactly $k$ voters that prefer $a$ to $p$. 
On the other hand, in $E$ there are $N+k-2$ voters that prefer $b$ to $p$
and $N-k-1$ voters that prefer $p$ to $b$. Thus, we have to replace
exactly $k$ voters that prefer $b$ to $p$.
Now, there is no voter among the last $N-3$ voters 
that ranks both $a$ and $b$ above $p$, which proves
our claim.
Now, consider a candidate $y_i$, $i=1, \dots, M$. Among the first $N$
voters, there are exactly two voters (corresponding to the endpoints 
of the edge $y_i$) that prefer $y_i$ to $p$. Hence, altogether there
are $N-1$ voters that prefer $p$ to $y_i$ and $N-2$ voters
that prefer $y_i$ to $p$. Thus, for every candidate $y_i$ we have to replace
at least one voter that ranks him above $p$, and such a voter corresponds
to an endpoint of $y_i$. Hence, the set of replaced voters 
directly corresponds to a vertex cover of $\Gamma$.
Similarly, suppose that $X'\subset X$, $|X'|\le k$,
is a vertex cover for $\Gamma$. Then by replacing the corresponding voters
with voters that rank $p$ first
we can ensure that $p$ beats all candidates in $Y$.
Clearly, $p$ also beats $z$. Finally, if $|X'|<k$, we 
replace another $k-|X'|$ of the first $N$ voters
with voters that rank $p$ first.
After this step, $p$ beats $a$, $b$, and $c$, 
so he becomes the Condorcet winner after at most
$|X'|+k-|X'|=k$ voter replacements.

It remains to show that we can ensure that
none of the remaining candidates is close to being 
a Condorcet winner. For $a$, $b$, and $c$ this is easy to achieve.
Set $t=2N-3$, and
require that at least $t/3$ voters prefer $a$ to $b$
to $c$, at least $t/3$ voters prefer $b$ to $c$
to $a$, and at least $t/3$ voters prefert $c$ to $a$
to $b$ (recall that by our assumption $N$ is divisible by $3$).
This ensures that $a$, $b$, and $c$ prevent each other from becoming
the Condorcet winners: indeed, at least $2t/3$ voters
prefer $a$ to $b$, at least $2t/3$ voters
prefer $b$ to $c$, and at least $2t/3$ voters
prefer $c$ to $a$, so the replacement score of each of these candidates is at least 
$\lceil t/6\rceil > k$. 

We use a similar construction for the candidates in $Y$.
Specifically, if $2N-3\le M$, we would like the $i$-th voter, 
$i=1, \dots, 2N-3$, 
to have a preference ordering (as restricted to $Y$) given by
\begin{equation}\label{star}
y_{M-i+2}\succ\dots\succ y_{M}\succ y_1\succ\dots\succ y_{M-i+1}, 
\end{equation} 
where we identify $y_{M+j}$ with $y_{j}$.
If $2N-3>M$, we would like to divide the voters into
$\lceil\frac{2N-3}{M}\rceil$ groups, where the first
$\lfloor\frac{2N-3}{M}\rfloor$ groups have size $M$, 
and the remaining group has size at most $M$, so that the $i$-th
voter in each group has preference ordering given by~\eqref{star}.
Since $M\ge N/2$, there will be at most four groups. 
Under this preference profile, which we will denote by 
$O^*=(o^*_1, \dots, o^*_{2N-3})$, 
for each $j=1, \dots, M$
there are at most four voters that rank $y_j$ above $y_{j-1}$, 
i.e., the replacement score of each $y\in Y$ is at least $N-5>k$.
However, this conflicts with the requirement that 
for $i=1,\dots, N$ the $i$-th voter
prefers candidates in $Y_i$ to those in $Y\setminus Y_i$.
Thus, we require that his preferences are given by $o^*_{i}$ 
insomuch as this is possible, i.e., he ranks the candidates
within $Y_i$ and $Y\setminus Y_i$ according to $o^*_i$, 
but ranks all candidates in $Y_i$ above those in 
$Y\setminus Y_i$.
Also, for $i=N+1, \dots, 2N-3$ we require voter $i$
to rank the candidates in $Y$ according to $o^*_i$.
Now, for each $y_j$ there are at most two sets $Y_i$
such that $y_j\in Y_i$. 
It follows that for each $j=1, \dots, M$, at most 
six voters prefer $y_j$ to $y_{j-1}$ (where $y_0=y_{M}$), 
so none of the candidates in $Y$ is close to being the Condorcet winner.
\end{proof}

For Young's rule, the winner determination problem is known to be
complete for the complexity class
$\Theta_2^p$~\cite{rot-spa-vog:j:young}.  It seems likely that this is
also the case for the voter replacement rule.

Observe that out of the three voting rules considered in this section,
one (Maximin) has an efficient winner determination procedure, while
the other two do not (assuming, as is currently believed, that no
$\np$-hard problem can be solved efficiently---i.e., in polynomial
time---for all instances).  The intuitive reason for this difference
is that
when we add voters 
to make a
candidate $c$ the Condorcet winner, we only need to add voters that
rank $c$ first, and, moreover, it does not matter how these voters
rank other candidates. On the other hand, when we delete or replace
voters, we have to choose which voters to remove, and this decision is
not straightforward.

\section{Conclusions and Future Research}
\label{sec:conclusions}


We have shown that two classical voting rules, Young's rule and 
Maximin, can be distance rationalized by the Condorcet consensus class and distances
of Hamming type. This further advances the project of classifying common voting rules by a consensus class and a distance
(only some multistage elimination rules are now left without known distance rationalizations). 
We have also shown that a the existing distance rationalization of Young's rule in fact leads to a somewhat different rule. 

Now the question of quality of such rationalization comes to the fore. Indeed, some distances are more natural than others so are the consensus classes. In particular, Kendall tau distance\footnote{also known as Dodgson distance, bubble-sort distance, swap distance, edit distance etc.} seems 
a particularly natural one and it is employed in distance-rationalizations of many rules.
So are the unanimity consensus and Condorcet consensus classes.  On the other hand,  \citeA{elk-fal-sli:c:votewise-dr} have shown that any rule can be distance rationalized if  unnatural consensus classes or unnatural distances are allowed. 
\citeA{elk-fal-sli:c:votewise-dr} identified  a certain family of distances,  called votewise distances (for example, Kendall tau and Hamming distances are votewise), as those that are
particularly natural. Not all voting rules can be rationalized with the use of those distances, in particular, STV cannot.  It is thus interesting  if it is possible to rationalize Young's rule and Maximin using votewise distances.

Another natural research direction is  to seek further connections between distance
rationalizability and  maximum likelihood
estimation approaches (see, e.g., ~\citeA{con-san:c:likelihood-estimators,con-rog-xia:c:mle}).

\section{acknowledgements}
  Some of the results of this paper were previously presented at the
  12th Conference on Theoretical Aspects of Rationality and Knowledge
  (TARK-09) under the title ``On Distance Rationalizability of Some
  Voting Rules'', and the authors would like to thank the anonymous
  TARK referees for their thorough and tremendously helpful work.
  The authors are also grateful to Felix Brandt, Paul Harrenstein, 
  Jerome Lang, and Alexis Tsoukias for useful discussions, feedback, 
  and pointers to the literature.   
  Piotr Faliszewski was supported by AGH University of Science and
  Technology Grant no. 11.11.120.865, Foundation for Polish
  Science's program Homing/Powroty, and Polish Ministry of
  Science and Higher Education grant N-N206-378637. Edith Elkind was supported by
  ESRC (grant no. ES/F035845/1), EPSRC (grant no. GR/T10664/01), and NRF
  Research Fellowship (NRF-RF2009-08). Arkadii Slinko was supported by the Science Faculty of the University of Auckland FRDF grant 3624495/9844.

\bibliography{grypiotr2006}

\appendix

\section{Computational Complexity Preliminaries}
In this section we give a brief review of the notions from
computational complexity theory.  The readers interested in a more in-depth
treatment of computational complexity are pointed to the classic
textbooks of~\citeA{gar-joh:b:int} and~\citeA{pap:b:complexity}.\par\medskip

Computational complexity theory is a branch of theoretical computer
science whose goal (or, rather, one of many goals) is to classify
computational problems with respect to the amount of resources they
require for obtaining solutions. The most standard resource type that
complexity theorists study is the time (i.e., the number of basic
computational steps) needed to find a solution.

Typically, instead of studying problems that ask one to compute some
mathematical object or to optimize some function, computational
complexity theory focuses on decision problems, that is, problems with
a \emph{yes/no} answer.  Decision problems are easier to work with
and, in most cases, preserve the resource requirements of the more
involved problems they are based on.  For instance, given an election $E =
(C,V)$ and some voting rule $\R$, instead of asking ``who won election
$E$ according to $\R$'' we can ask, for each $c \in C$, ``did $c$ win
election $E$ according to $\R$.''

One of the most crude, but at the same time very practical and
natural, ways to classify decision problems is to classify them as
either belonging to the class $\p$ (that is, the class of problems
that can be solved in polynomial time), or being $\np$-hard.  We will
explain what it means for a problem to be $\np$-hard in the next
paragraph. The classification is not perfect as some problems are
neither in $\p$ nor are $\np$-hard, but most problems encountered in
practice indeed fall into one of these two groups. The
problems in $\p$ are considered computationally easy because given an
instance $I$ of a problem from $\p$, it is possible to solve it using
at most polynomially many steps (with respect to the number of bits
needed to encode $I$).  On the other hand, it is widely believed that
if a problem is $\np$-hard then, in general, to solve its instance $I$
one needs to make at least an exponential number of steps. In the next
section we will show that, indeed, the problem of deciding whether a
given candidate $c$ is a winner with respect to the voter replacement
rule is $\np$-hard and, as a result, that it is computationally
difficult (as least as long as $\p \neq \np$; which is a widely
believed conjecture).

How is the notion of $\np$-hardness defined?  To answer this question
we need to describe the class $\np$ first.  However, instead of
defining the class formally, we find it more practical to provide the
intuition behind the class and point the readers to the classic texts
of~\citeA{gar-joh:b:int} and~\citeA{pap:b:complexity} for technical
details. Let us start with the following problem as an example.

\begin{definition}\label{def:vertex-cover}
  An instance of {\sc Vertex Cover} is given by a pair $(\Gamma=(X,
  Y); k)$ where $\Gamma$ is a graph with a vertex set $X$ and an edge
  set $Y$, and $k\in{\mathbb N}$. It is a ``yes''-instance if $\Gamma$
  has a vertex cover of size at most $k$ (i.e., if it is possible to
  pick $k$ vertices such that each edge is incident to at least one of
  the selected vertices), and a ``no''-instance otherwise.
\end{definition}

In {\sc Vertex Cover} we ask whether a subset of vertices with a certain
property exists. It is not at all clear how to compute such a set
efficiently (that is, in polynomial time) but, if we were given some
set of vertices, we could easily verify if it indeed satisfies our
requirements. Namely, we would check if every edge is incident to at
least one of the vertices from the set and if the set contains at most
$k$ elements.  Thus, while it seems computationally hard to solve {\sc
  Vertex Cover}, it is very easy to verify if a solution provided by
someone else is correct.\footnote{Note that here by ``solution'' we do
  not mean the \emph{yes/no} answer but rather the underlying
  mathematical object the decision problems asks about.}  Now, the
class $\np$ is exactly the class of problems for which, given an
instance $I$ and a solution $s$ for it, it is possible to verify the
solution $s$ in time polynomial in the number of bits encoding $I$.  A
decision problem $A$ is called $\np$-hard if it is at least as hard as
the hardest problem in $\np$. To formalize the notion of ``is at least
as hard as'' we use many-one polynomial-time reductions.

\begin{definition}
Let $A$ and $B$ be two decision problems. We say that $A$ many-one
reduces to $B$ in polynomial time if there exists a function $f$
such that:
\begin{enumerate}
\item $f$ is computable in polynomial time, and
\item for each instance $x$ of the problem $A$ it holds that the
  answer for $x$ is ``yes'' if and only if the answer for $f(x)$ is
  ``yes.''
\end{enumerate}
\end{definition}
In other words, if a problem $A$ many-one polynomial-time reduces to
a problem $B$, then it is easy to translate questions in the format of
$A$ to the questions in the format of $B$, while preserving the
answers. A decision problem $A$ is called $\np$-hard if all problems
in $\np$ many-one reduce to it in polynomial time. If $A$ is both
$\np$-hard \emph{and} a member of $\np$ then $A$ is called
$\np$-complete.  While at first it may seem that $\np$-complete and
$\np$-hard problems might not even exist, in fact many hundreds of
natural $\np$-complete problems have been identified (see the text
of~\citeA{gar-joh:b:int} for a very early list).  For example, {\sc
  Vertex Cover} is $\np$-complete.

The next proposition is the standard tool for proving that
a given problem is $\np$-hard. To show that a problem is $\np$-hard
it suffices to show that some previously known $\np$-hard
problem many-one reduces to it in polynomial time.

\begin{proposition}
  Let $A$ be a decision problem and let $B$ be some known $\np$-hard
  problem. If $B$ many-one reduces to $A$ in polynomial time then $A$
  is $\np$-hard.
\end{proposition}

This proposition follows immediately from the observation that the
relation ``polynomial-time many-one reduces to'' is transitive.  We
are now ready to prove Theorem~\ref{thm:nphard}.

\end{document}